\documentclass[conference]{IEEEtran}

\usepackage{blindtext, graphicx}
\usepackage{amsmath}
\usepackage{amsthm}

\usepackage{tikz}

\usetikzlibrary{shapes,arrows}
\usepackage{array, tabularx}
\usepackage{multirow,multicol, makecell, booktabs, caption}

\usepackage{pifont}
\usepackage{array,tabularx}
\usepackage{lipsum}

\usepackage{xcolor}
\usepackage{caption}
\usepackage{subcaption}
\usepackage{authblk}
\usepackage{bbm}
\usepackage{etoolbox}%
\usepackage[utf8]{inputenc}
\usepackage[english]{babel}

\newtheorem{theorem}{Theorem}

\usepackage[colorinlistoftodos]{todonotes}
\newcolumntype{P}[1]{>{\centering\arraybackslash}p{#1}}
\newcolumntype{C}[1]{>{\centering\arraybackslash}m{#1}}
\usepackage{amstext}  
\usepackage{footnote}
\makesavenoteenv{tabular}
\makesavenoteenv{table}

\usepackage{times}
\usepackage{fancyhdr,graphicx,amssymb}
\usepackage[linesnumbered,ruled,vlined]{algorithm2e}
\usepackage{tikz}

\usepackage{float}
\floatstyle{plaintop}
\restylefloat{table}
\graphicspath{ {./images/} }
\graphicspath{{Pictures/}} 
\usepackage{tikz}
\newcommand{\rectanglesign}[1]{
    \mathbin{
        \mathchoice
        {\buildrectanglesign{\displaystyle}{#1}}
        {\buildrectanglesign{\textstyle}{#1}}
        {\buildrectanglesign{\scriptstyle}{#1}}
        {\buildrectanglesign{\scriptscriptstyle}{#1}}
    }
}

\newcommand\buildrectanglesign[2]{%
    \begin{tikzpicture}[baseline=(X.base), inner sep=1, outer sep=0]
    \node[draw,rectangle] (X)  {\ensuremath{#1 #2}};
    \end{tikzpicture}%
}

\newcommand{\circlesign}[1]{
    \mathbin{
        \mathchoice
        {\buildcirclesign{\displaystyle}{#1}}
        {\buildcirclesign{\textstyle}{#1}}
        {\buildcirclesign{\scriptstyle}{#1}}
        {\buildcirclesign{\scriptscriptstyle}{#1}}
    }
}

\newcommand\buildcirclesign[2]{%
    \begin{tikzpicture}[baseline=(X.base), inner sep=0.1, outer sep=0]
    \node[draw,circle] (X)  {\ensuremath{#1 #2}};
    \end{tikzpicture}%
}

\SetCommentSty{mycommfont}
\include{pythonlisting}

\begin{document}
\title{Polar Coded Repetition for Low-Capacity Channels}




\author{
Fariba Abbasi$^\dag$, Hessam Mahdavifar$^\ddag$, and Emanuele Viterbo$^\dag$\\
$^\dag$ Monash University, Melbourne, VIC3800, Australia\\
$^\ddag$ University of Michigan, Ann Arbor, MI 48104, USA\\
Email: $^\dag$\{fariba.abbasi, emanuele.viterbo\}@monash.edu, $^\ddag$ hessam@umich.edu
}

\maketitle

\begin{abstract}
Constructing efficient low-rate error-correcting codes with low-complexity encoding and decoding have become increasingly important for applications involving ultra-low-power devices such as Internet-of-Things (IoT) networks. To this end, schemes based on concatenating the state-of-the-art codes at moderate rates with repetition codes have emerged as practical solutions deployed in various standards. In this paper, we propose a novel mechanism for concatenating outer polar codes with inner repetition codes which we refer to as \textit{polar coded repetition}. More specifically, we 
propose to transmit a slightly modified polar codeword by deviating from Ar{\i}kan's standard $2 \times 2$ Kernel in a certain number of polarization recursions at each repetition block. We show how this modification can improve the asymptotic achievable rate of the polar-repetition scheme, while ensuring that the overall encoding and decoding complexity is kept almost the same. The    achievable rate is analyzed for the binary erasure channels (BEC).
\end{abstract}

\section{Introduction}

Recently, the Third Generation Partnership Project (3GPP) has introduced various features including Narrow-Band Internet of Things (NB-IoT) and enhanced Machine-Type Communications (eMTC) into the cellular standard in order to address the diverse requirements of massive IoT networks including low-power and wide-area (LPWA) cellular connectivity \cite{RRatasuk}. 

In general, devices in IoT networks have strict limitations on their total available power and are not equipped with advanced transceivers due to cost constraints. Consequently, they often need to operate at very low signal-to-noise ratio (SNR) necessitating ultra-low-rate error-correcting codes for reliable communications. For instance, the SNR of $-13$ dB is translated to capacity being $0.03$ bits per transmission. The solution adopted in the 3GPP standard is to use the legacy turbo codes or convolutional codes at moderate rates, e.g., the turbo code of rate $1/3$, together with up to $2048$ repetitions to support effective code rates as low as $1.6 \times 10^{-4}$. Although this repetition leads to efficient implementations with reduced computational complexity, repeating a high-rate code to enable low-rate communication will result in rate loss and mediocre performance. As a result, studying efficient channel coding strategies for reliable communication in this low SNR regime, where channel coding is the only choice, is necessary \cite{MFereydounian}.

The fundamental non-asymptotic laws for channel coding in the low-capacity regimes have been recently studied in \cite{MFereydounian}. Furthermore, the optimal number of repetitions with negligible rate loss, in terms of the code block length and the underlying channel capacity, is characterized in \cite{MFereydounian}. It is also shown in \cite{MFereydounian} that the state-of-the-art polar codes, proposed by Ar{\i}kan \cite{Arikan}, naturally invoke
this optimal number of repetitions when constructed for low-capacity channels. In another related work, low-rate codes for binary symmetric channels are constructed by concatenating high-rate i.e., rate close to $1$, polar codes with repetitions \cite{Dumer}.




In this paper, we propose an alternative mechanism called {\em coded repetition,}  for the repetition concatenation scheme. 
A slightly modified codeword in each repetition block is transmitted instead of identical codewords in all repetition blocks. The goal is to reduce the rate loss due to the repetition while keeping the overall encoding and decoding complexity the same as in a standard repetition concatenation scheme. 
In particular, we consider polar codes as the outer code. In the proposed polar coded repetition scheme, a slightly modified polar codeword is transmitted in each repetition block by deviating from Ar{\i}kan's standard $2 \times 2$ Kernel in a certain number of polarization recursions at each repetition block. We show that our proposed scheme outperforms the straightforward polar-repetition scheme, in terms of the asymptotic achievable rate, for any given number of repetitions over the binary erasure channel (BEC). The proposed polar coded repetition has almost the same encoding and decoding complexity as the straightforward repetition scheme.



\subsection{Background}
Consider two copies of a binary discrete memoryless channel (B-DMC) $W:\mathcal{X} \rightarrow \mathcal{Y}$ with binary inputs $x_1$, $x_2$ $\in \mathcal{X}$ and outputs $y_1$, $y_2$ $\in \mathcal{Y}$. The transformation $G_2=\footnotesize \begin{pmatrix}
1 & 0\\
1 & 1
\end{pmatrix}$ is applied on the inputs of these two channels and $u_1$ and $u_2$ are generated. Then, $x_1$ and $x_2$ are transmitted through the independent copies of $W$. At the decoder side, $u_1$ is decoded by using two observations $y_1, y_2$ and then $u_2$ is decoded by using the decoded sequence, $\hat{u}_1$, and the observations $y_1, y_2$. The transformation $G_2$ along with this successive decoding, referred to as {\em successive cancellation (SC)}, transforms the two copies of the channel $W$ into two synthetic channels $W^0: W \rectanglesign{*} W: \mathcal{X} \rightarrow \mathcal{Y}^2$ and $W^1: W \circlesign{*} W: \mathcal{X} \rightarrow \mathcal{Y}^2 \times \mathcal{X}$ as follows:
\begin{equation}
\begin{aligned}
W \rectanglesign{*} W (y_1,y_2|u_1)&=\sum_{u_2 \in \mathcal{X}} \frac{1}{2} W(y_1|u_1+u_2) W(y_2|u_2), \\
W \circlesign{*} W (y_1,y_2,u_1|u_2)&=\frac{1}{2} W(y_1|u_1+u_2) W(y_2|u_2).
\end{aligned}
\end{equation}
Here, the channel $W^0$ is \textit{weaker} (i.e., less reliable) compared to $W$, while the channel $W^1$ is \textit{stronger} (i.e., more reliable) compared to the channel $W$.
The quality of a channel is measured by a reliability metric such as the Bhattacharyya parameter defined as 
\begin{equation}
Z(W)\overset{\Delta}{=}\sum_{y \in \mathcal{Y}} \sqrt{W(y|0) W(y|1)},
\end{equation}
which is equal to the erasure probability for BECs, i.e., for BEC($\epsilon$), $Z(W)=\epsilon$. The Bathacharyya parameters of the synthetic channels follow the properties
\begin{equation}
\begin{aligned}
Z(W^1)&= Z(W)^2,\\
Z(W^0)&\leq 2 Z(W)-Z(W)^2, \label{Bathacharyya}
\end{aligned}
\end{equation}
 with equality in (\ref{Bathacharyya}) iff $W$ is a BEC channel. 

If we continue applying the transformation $G_2$ recursively $m$ times, we will obtain $n=2^m$ synthetic channels $\{W_m^{(i)}\}_{i \in \{0, 1, \ldots, n-1\}}$.
More specifically, if we let $\{i_1, i_2, ..., i_m\}$ be the binary expansion of $i=\{0, 1, ..., n-1\}$ over $m$ bits, where $i_1$ is the most significant bit and $i_m$ is the least significant one, then we define the synthetic channels $\{W_m^{(i)}\}_ {i\in \{0,..., n-1\}}$ as
\begin{equation}
W_m^{(i)}=(((W^{i_1})^{i_2})^{...})^{i_m}.\label{synthetic channel} 
\end{equation}

Ar\i kan in his seminal paper, \cite{Arikan}, showed that as $m \rightarrow \infty$, these $2^m$ synthetic channels are either purely noiseless or purely noisy channels. 
Thus, on the encoder side, using $k$ entries of the input vector $u_0^{n-1}$ as the information bits and setting the remaining entries to zero (frozen bits) will provide almost error-free communication. 
Hence, an $(n=2^m, k)$ polar code is a linear block code generated by $k$ rows of $G_n=G_2^{\otimes m}$,  which correspond to the best $k$ synthetic channels. Here, $.^{\otimes m}$ is the $m$-times Kronecker product of a matrix with itself.

Repetition coding is a simple way of designing a practical low-rate code. 
Let 
$r$ 
denote the number of the repetitions and $N$, the length of the code. For constructing the repetition code, first, one needs  to design a smaller outer code (e.g. polar codes) of length $n=N/r$ for channel $W^r$ and then repeat each of its code bits $r$ times. Consequently, the length of the final code will be $n \times r = N$. 
This is equivalent to transmitting an input bit over the $r$-repetition channel $W^r$ and outputs an $r$ tuple. 
For example, if $W$ is BEC$(\epsilon)$, then its corresponding $r$-repetition channel is $W^r = BEC(\epsilon^r)$. 
The main advantage of this concatenation scheme is that the decoding complexity is essentially reduced to that of the outer code making it appealing to low-power applications. This comes at the expense of loss in the asymptotic achievable rate especially if the number of the repetitions is large. Suppose that $C(W)$ is the capacity of the channel $W$ and $N C(W)$ is the capacity corresponding to $N$ channel transmissions. With repetition coding, since we transmit $n$ times over the channel $W^r$, the capacity will be reduced to $n C(W^r)$. Note that, in general, we have $n C(W^r) \leq N C(W)$ and the ratio vanishes with growing $r$. Let's consider BEC($\epsilon$) as an example with $r=2$. If $\epsilon=0.5$, then $\frac{1}{2}C(W^2)=0.375$ whereas $C(W)=0.5$. However, when $\epsilon$ is close to 1, $C(W^2)=1-\epsilon^2$ is very close $2C(W)=2(1-\epsilon)$.

\section{Proposed Scheme}
In this section, the proposed polar coded repetition scheme is discussed. It is shown how to improve the performance of the straightforward repetition scheme in the low-rate regime, while keeping the computational complexity almost the same as the original one.

Consider an outer polar code with $r=2^t$ repetitions and let $c$ denote a polar codeword of length $n=2^m$ designed for transmission over a channel $W$, $r$ times. 
Owing to the recursive structure of the polar codes, one can write the polarization transform matrix as $G_n=G^{\prime}_{r^ \prime} \otimes G_2^{\otimes (m-t^{\prime})}$, where $G^{\prime}_{r^ \prime}$ is an $r^{\prime} \times r^{\prime}$ binary matrix with $r^{\prime}=2^{t^{\prime}}$. 
In our proposed scheme, we consider a different $G^{\prime}_{r^{\prime}}$ in each repetition block, while keeping $G_2^{\otimes (m-t^{\prime})}$ the same in all of them. In other words, the first $t'$ recursions of Ar{\i}kan's polarization transform are modified in each repetition while the rest of $m-t'$ recursions are kept the same.
Note that if one chooses $r^{\prime} = n$, i.e., the transmission in each block being different, then the channel capacity $C(W)$ can be achieved. 
However, we choose $r^{\prime}=r$ to have a comparable complexity with the straightforward polar-repetition scheme. The complexity of the simple polar-repetition and the proposed modified polar-repetition schemes will be provided in subsection \ref{Complexity}. 


We illustrate the idea through some examples with two and four repetitions and constructed with \textit{regular} and \textit{irregular} polar coding approaches. Then, we generalize the regular scheme to accommodate an arbitrary repetition $r$.  

\subsection{Examples for two and four repetitions}
    In this subsection, we provide three examples for two and four repetitions as follows.
    
    \textbf{Example 1 (Two repetitions):} Consider an outer polar code with two repetitions. 
    Hence, the polar codeword $c$ needs to be designed for $W^2= W \circlesign{*} W$. The recursive structure of polar codes implies that codeword $c= (c_1 \oplus c_2, c_2)$ is constructed from the generator matrix $G_n=G^{\prime}_2 \otimes G_2^{\otimes (m-1)}$, where $G^{\prime}_2=\footnotesize \begin{pmatrix}
1 & 0\\
1 & 1
\end{pmatrix}$ and $c_1$ and $c_2$ are polar codewords of length $n/2$ generated from $G_2^{\otimes (m-1)}$. 
    
Now, we consider an alternative scheme where in each repetition, we transmit different combinations of $c_1$ and $c_2$ by choosing different $G^{\prime}_2$ in each of them. Let ${G^{\prime}}^{(i)}_2$ be a lower triangular matrix\footnote { \cite{Fazeli} showed that the column permutations and the one-directional row operations can always transform a non-singular kernel ${G^{\prime}}^{(i)}_2$ to a lower triangular kernel $G^{\prime\prime}$ with the same polarization behavior.} ${G^{\prime}}^{(i)}_2=\footnotesize \begin{pmatrix}
1 & 0\\
e & 1
\end{pmatrix}$, where $e \in \mathcal{F}_2$ and $i=\{1,2\}$ is the index of the transmission (see TABLE. \ref{2transmision} for two possible matrices). 
  \begin{table}[h]
\begin{center}
\scalebox{0.9}{
\begin{tabular}{|P{1.6cm}|P{1.5cm}|}
 \hline
 Pattern no. & ${G^{\prime}}^{(i)}_2$ \\
\hline
 $P^{(0)}_2 $ &  $\footnotesize \begin{pmatrix} 1 & 0\\ 1 & 1
 \end{pmatrix}$ \\  
\hline
 $P^{(1)}_2$ &   $\footnotesize \begin{pmatrix} 1 & 0\\ 0 & 1
 \end{pmatrix} $  \\ 
\hline
 
\end{tabular}%
}
\end{center}
\caption{Two possible matrices for two repetitions}
\label{2transmision}
\vspace{-6mm}
\end{table}
     There are three possible cases for two transmissions as follows.
     \begin{enumerate}
    \item  ${G^{\prime}}^{(1)}_2=\footnotesize \begin{pmatrix}
1 & 0\\
1 & 1
\end{pmatrix}$ and ${G^{\prime}}^{(2)}_2=\footnotesize \begin{pmatrix}
1 & 0\\
1 & 1
\end{pmatrix}$: In this case, $(c_1 \oplus c_2, c_2)$ and $(c_1 \oplus c_2, c_2)$ are transmitted in each repetition. By considering both transmissions, one concludes that codeword $c_1$ is implicitly designed for the effective channel that the sub-block of length $n/2$ observes, i.e., for $W^2 \rectanglesign{*} W^2$ and $c_2$ is designed for $W^2 \circlesign{*} W^2$. As a result, the capacity per channel use per transmission for this case and specifically for BEC will be
\begin{equation*}
 \begin{aligned}
C^{(1)}_{2}& =(C(W^2 \rectanglesign{*} W^2)+C(W^2 \circlesign{*} W^2))/4\\
&= (1-\epsilon^2)/2.
\end{aligned}
\end{equation*}

\item ${G^{\prime}}^{(1)}_2=\footnotesize \begin{pmatrix}
1 & 0\\
1 & 1
\end{pmatrix}$ and ${G^{\prime}}^{(2)}_2=\footnotesize \begin{pmatrix}
1 & 0\\
0 & 1
\end{pmatrix}$: For this case, $(c_1 \oplus c_2, c_2)$ and $(c_1, c_2)$ are transmitted in the first and second repetitions. Codeword $c_1$ is designed for the effective channel that the sub-block of length $n/2$ observes, i.e., for $(W \rectanglesign{*} W^2) \circlesign{*}  W$, and $c_2$ is designed for $W^2 \circlesign{*} W$. As a result, the capacity per channel use per transmission for this case is
\begin{equation*}
 \begin{aligned}
C^{(2)}_{2}& =(C((W \rectanglesign{*} W^2) \circlesign{*} W)+C(W^2 \circlesign{*} W))/4\\
&= (2-\epsilon^2-2\epsilon^3+\epsilon^4)/4.
 \end{aligned}
\end{equation*}

\item ${G^{\prime}}^{(1)}_2=\footnotesize \begin{pmatrix}
1 & 0\\
0 & 1
\end{pmatrix}$ and ${G^{\prime}}^{(2)}_2=\footnotesize \begin{pmatrix}
1 & 0\\
0 & 1
\end{pmatrix}$: In the first and second repetitions, $(c_1, c_2)$ and $(c_1, c_2)$ are transmitted. Both Codewords $c_1$ and $c_2$ are designed for the effective channel that the sub-block of length $n/2$ observes, i.e., for $W^2$. As a result, the capacity for this case will be
\begin{equation*}
 \begin{aligned}
C^{(3)}_{2}& =(C(W^2)+C(W^2))/4\\
&= (1-\epsilon^2)/2.
 \end{aligned}
\end{equation*}
\end{enumerate}
It can be observed that for $0<\epsilon<1$, the capacity of case 2 is larger than the capacities of both cases 1 and 3, which are simple repetition schemes. In other words,
    \begin{equation}
       C((W \rectanglesign{*} W^2) \circlesign{*} W)+C(W^3) > 2 C(W^2),
     \label{comparison}
    \end{equation}
where the right hand side of (\ref{comparison}) is the capacity for the straightforward repetition scheme and the left hand side of (\ref{comparison}) is the capacity of case 2.

In the proposed modified approach, which we refer to as coded repetition scheme, we consider case 2. This modified scheme has the same encoding/decoding complexity
compared to a simple repetition scheme.

\textbf{Example 2 (Four repetitions with regular polar codes):} Consider an outer polar codes with four repetitions. Since we intend to keep the complexity of the proposed scheme the same as the complexity of the simple repetition one, let's consider all possible Kronecker products of the patterns $P^{(0)}_2$ and $P^{(1)}_2$ for ${G^{\prime}}^{(i)}_{4R}$, $i=\{1, 2, 3, 4\}$ as the ones depicted in Table \ref{4transmissions}. We call these patterns {\em regular} polar codes.
\begin{table}
\begin{center}
\scalebox{0.8}{
\begin{tabular}{|P{1.6cm}|P{3.5cm}|}
 \hline
  Pattern no. & ${G^{\prime}}^{(i)}_{4R}$ \\
\hline
 $P^{(0)}_{4R} $ & $\footnotesize \begin{pmatrix} 1 & 0\\ 1 & 1
 \end{pmatrix} \otimes \footnotesize \begin{pmatrix} 1 & 0\\ 1 & 1
 \end{pmatrix}$ \\  
\hline
 $P^{(1)}_{4R}$ &  $\footnotesize \begin{pmatrix} 1 & 0\\ 1 & 1
 \end{pmatrix} \otimes \footnotesize \begin{pmatrix} 1 & 0\\ 0 & 1
 \end{pmatrix}$  \\ 
\hline
 $P^{(2)}_{4R}$ & $\footnotesize \begin{pmatrix} 1 & 0\\ 0 & 1
 \end{pmatrix} \otimes \footnotesize \begin{pmatrix} 1 & 0\\ 1 & 1
 \end{pmatrix}$  \\ 
 \hline
 $P^{(3)}_{4R}$ &  $\footnotesize \begin{pmatrix} 1 & 0\\ 0 & 1
 \end{pmatrix} \otimes \footnotesize \begin{pmatrix} 1 & 0\\ 0 & 1
 \end{pmatrix}$ \\ 
\hline
\end{tabular}%
}

\end{center}
\caption{All possible cases for four repetitions}
\label{4transmissions}
\vspace{-8mm}
\end{table} 
Then, for four transmissions, we try all $35$ multi-subsets of size $4$ from the set $\{P^{(0)}_{4R}, P^{(1)}_{4R}, P^{(2)}_{4R}, P^{(3)}_{4R}\}$ to find the best one in terms of the capacity. The channel that each codeword $c_i$ observes follows the recursive structure shown in Fig. \ref{fig:P4}. 
\begin{figure*}[h]
\centering
\includegraphics[width=0.78\linewidth]{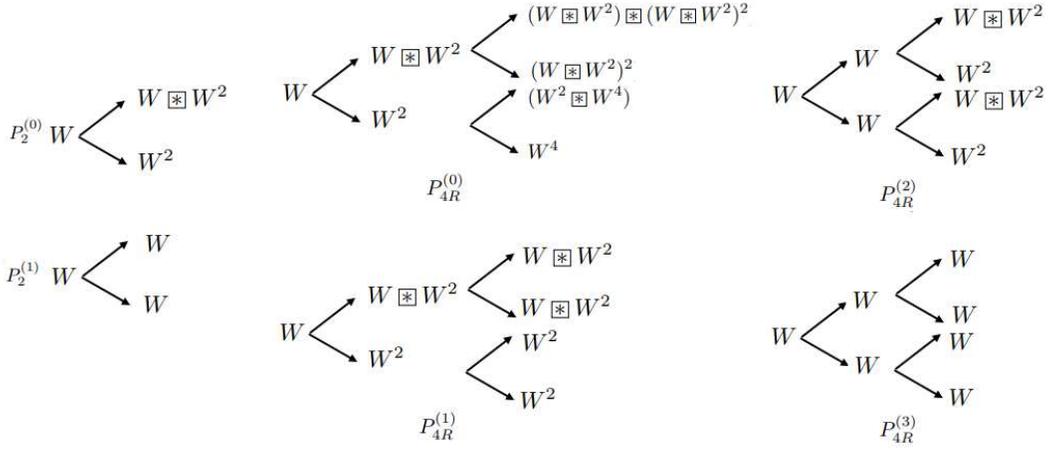}
\caption{The recursive structure of the channels that each codeword $c_i$ observes for two and four transmissions.}
\label{fig:P4}
\end{figure*}
With a simple search among these $35$ multi-subsets, it is found that the pattern $(P^{(0)}_{4R}, P^{(3)}_{4R}, P^{(3)}_{4R}, P^{(3)}_{4R})$ has the largest capacity.

In this modified repetition scheme, $ (c_1 \oplus c_2 \oplus c_3 \oplus c_4, c_2 \oplus c_4, c_3 \oplus c_4, c_4), (c_1, c_2, c_3, c_4), (c_1, c_2, c_3, c_4)$ and $(c_1, c_2, c_3, c_4)$ are transmitted in the first, second, third and fourth transmissions, respectively.
Codword $c_1$ is constructed for the effective channel that the first sub-block of length $n/4$ observes, i.e., for $W_1=((W \rectanglesign{*} {W^2}) \rectanglesign{*} (W \rectanglesign{*} {W^2})^2) \circlesign{*} W^3$, $c_2$  for $W_2=(W \rectanglesign{*} {W^2})^2 \circlesign{*} W^3$, $c_3$ for $W_3=({W}^2 \rectanglesign{*} {W}^4) \circlesign{*} W^3$ and $c_4$ for $W_4={W}^4 \circlesign{*} W^3$. For BEC $W$,  the capacity of the modified scheme is larger than that of the repetition scheme for $0<\epsilon<1$:
\begin{equation}
    C_{\text{4R}}= C(W_1)+C(W_2)+C(W_3)+C(W_4) > 4 C(W^4).
\end{equation}

\textbf{Example 3 (Four repetitions with irregular polar codes\footnote{
Note that regular scheme is a special case of the irregular scheme.}):}
We consider an alternative type of patterns for $4$ repetitions, referred to as {\em irregular} polar codes, which have the same computational complexity as the simple repetition scheme. These $8$ irregular patterns are constructed with ${G^{\prime}}^{(i)}_{4I}=\footnotesize \begin{pmatrix}
P^{(j)}_2 & 0\\
P^{(j)}_2 & P^{(j)}_2
\end{pmatrix}$ and ${G^{\prime}}^{(i)}_{4I}=\footnotesize \begin{pmatrix}
P^{(j)}_2 & 0\\
0 & P^{(j)}_2
\end{pmatrix}$, where $j=\{0,1\}$ and $i=\{1, 2, \ldots, 8\}$ (see Fig. \ref{fig:I41}).
\begin{figure*}
\vspace{-5mm}
\centering
\includegraphics[width=0.98\linewidth]{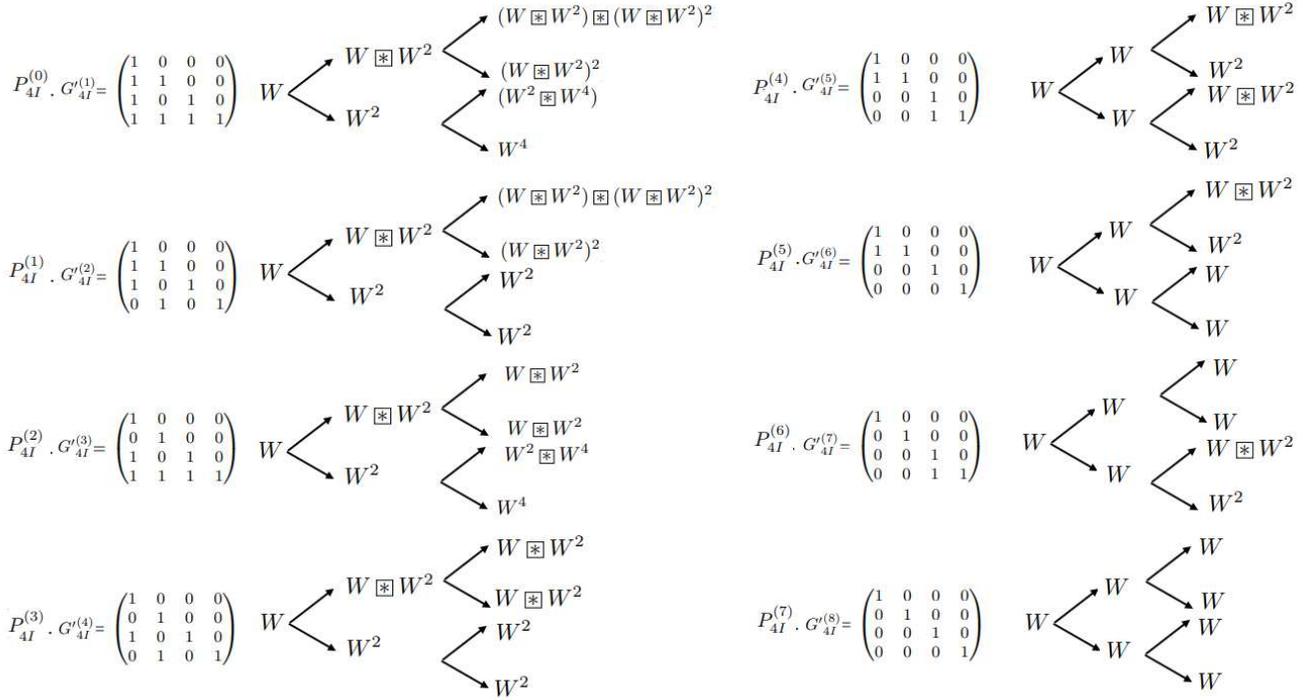}
\caption{All $8$ possible irregular kernels ${G^{\prime}}^{(i)}_{4I}$ for $4$ transmissions and the corresponding recursive structure of the channels that each codeword $c_i$ observes.}

\label{fig:I41}
\end{figure*}

With a simple search among all $330$ multi-subsets of size $4$ from the set $\{P^{(k)}_{4I}\}_{k=0}^{7}$,
it is found that the pattern $(P^{(2)}_{4I}, P^{(5)}_{4I}, P^{(7)}_{4I}, P^{(7)}_{4I})$ has the largest capacity. The channel that each codeword $c_i$ observes follows the recursive structure shown in Fig. \ref{fig:I41}.
In this scheme, $(c_1 \oplus c_3 \oplus c_4, c_2 \oplus c_4, c_3 \oplus c_4, c_4)$, $(c_1 \oplus c_2, c_2, c_3, c_4)$, $(c_1, c_2, c_3, c_4)$ and $(c_1, c_2, c_3, c_4)$ are transmitted in the first, second, third and fourth transmissions, respectively. Codeword $c_1$ is constructed for the effective channel $W_1=(W \rectanglesign{*} W^2) \circlesign{*} (W \rectanglesign{*} W^2) \circlesign{*} W \circlesign{*} W$, $c_2$ for $W_2=(W \rectanglesign{*} W^2) \circlesign{*} W^2 \circlesign{*} W \circlesign{*} W$, $c_3$ for $W_3=(W^2 \rectanglesign{*} W^4) \circlesign{*} W \circlesign{*} W \circlesign{*} W$ and $c_4$ for $W_4=W^4 \circlesign{*} W \circlesign{*} W \circlesign{*} W$. For BEC $W$, the capacity of the modified scheme with irregular polar codes is larger than the one with regular polar codes for $0<\epsilon<1$. In other words,
\begin{equation}
C_{\text{4I}}= C(W_1)+C(W_2)+C(W_3)+C(W_4) > C_{\text{4R}}.
\end{equation}

\subsection{General case for regular polar codes}
For the general case of $r=2^t$ repetitions with regular polar codes, we consider all $r$ possible $t$ times Kronecker products of the patterns $P_2^{(0)}$ and $P_2^{(1)}$, as $P^{(i)}_r$, $i=0, 1, \ldots, r-1$. In the proposed scheme, we use $P^{(0)}_r=(P_2^{(0)})^{\otimes t}$ for the first transmission and $P^{(r-1)}_{r}=(P_2^{(1)})^{\otimes t}$ for the rest $r-1$ ones.
For BEC $W$ with an erasure probability $\epsilon$, let's define ${Z_{P^{(i)}_r}(W^{(k)}_r)}\overset{\Delta}{=}{Z_{(i_1, \ldots, i_t)}(W^{(k)}_r)}$ as the erasure probabilities of the channels that each codeword $c_k$, $k=\{1,2, \ldots, r\}$ for pattern $P^{(i)}_r$ observes and $\{i_1, i_2, \ldots i_t\}$ as the $t$-bit binary expansion of $i$. 
Then, the recursive formula for computing $Z_{P^{(i)}_r}(W_r^{(k)})$ can be written as 
\begin{equation}
\begin{aligned}
    Z_{(i_1,\ldots, i_t)}&(W^{(2j-1)}_r)=Z_{(i_1, \ldots, i_{t-1})}(W^{(j)}_{\frac{r}{2}}) \times\\
    & [1 + Z_{(i_1, \ldots, i_{t-1})}(W^{(j)}_{\frac{r}{2}})-  Z_{(i_1, \ldots, i_{t-1})}^2(W^{(j)}_{\frac{r}{2}})]^{(1-i_t)},\\
    Z_{(i_1, \ldots, i_{t})}&(W^{(2j)}_r)=Z_{(i_1, \ldots, i_{t-1})}(W^{(j)}_{\frac{r}{2}}) \times\\
     & [Z_{(i_1, \ldots, i_{t-1})}(W^{(j)}_{\frac{r}{2}})]^{(1-i_t)},
    \end{aligned}
    \end{equation}
where $Z(W^{(1)}_1)=\epsilon$ and $j=1, 2, \ldots, \frac{r}{2}$.
Hence, the capacity for the proposed scheme will be
  \begin{equation}
\begin{aligned}
    C_{rR}&=\frac{r-\sum_{k=1}^r  Z_{P^{(0)}_r}(W^{(k)}_r) \times (Z_{P^{(r-1)}_{r}}(W^{(k)}_r))^{r-1}}{r^2}.
    \end{aligned}
    \end{equation}
    Since $Z_{P^{(r-1)}_{r}}(W^{(k)}_r)=\epsilon$, for all $k=1, 2, \ldots ,r$, we will have
    \begin{equation}
\begin{aligned}       
C_{rR}=\frac{r-\sum_{k=1}^r  Z_{P^{(0)}_r}(W^{(k)}_r) \times \epsilon ^{r-1}}{r^2}.
     \end{aligned}
 \end{equation}

Next, we show that $C_{rR} > \frac{C(W^r)}{r}$ for any $r$ repetitions and  $0<\epsilon<1$. In other words,
\begin{equation}
\begin{aligned}       
         \sum_{k=1}^r  Z_{P^{(0)}_r}(W^{(k)}_r) & < r \epsilon.\label{Theorem}
     \end{aligned}
 \end{equation}
To this end, we first prove that $\sum_{k=1}^r Z_{P^{(0)}_r}(W^{(k)}_r)- r \epsilon$ has zeros at $\epsilon=0$ and $\epsilon=1$.   

     
     

\begin{theorem}
$Z_{P^{(0)}_r}(W^{(k)}_r)=0$ at $\epsilon=0$ and $Z_{P^{(0)}_r}(W^{(k)}_r)=1$ at $\epsilon=1$ for all $k=\{1, 2, \ldots, r\}$. 
\end{theorem} 
 
 \begin{proof}
 Let us write the recursive formula for erasure probability as $Z_{P^{(0)}_r}(W_r^{(k)})=f_{k_1}(f_{k_2}(...f_{k_t}(\epsilon))))$, where $k_i=\{0, 1\}$, $i=\{1, 2, \ldots, t\}$ 
and $f_0(a)=a+a^2-a^3$, $f_1(a)=a^2$, $\forall k=\{1, 2, \ldots, r\}$.

Since $f_{k_{i}}(a)|_{a=1}=1$ and $f_{k_{i}}(a)|_{a=0}=0$, by using recursion, we conclude $Z_{P^{(0)}_r}(W^{(k)}_r) =1$ at $\epsilon=1$ and  $Z_{P^{(0)}_r}(W^{(k)}_r) =0$ at $\epsilon=0$ $\forall k=\{1, 2, \ldots r\}$. 
\end{proof}

Then, one can use Sturm algorithm\footnote{Although Sturm's theorem is a complete solution for finding the number of the real roots of the polynomials, when the degree of the polynomial increases, it isn't efficient in terms of implementation. The algorithm proposed in \cite{Vincent} is more efficient for higher degrees.} \cite{Sturm} to show that $\sum_{k=1}^r Z_{P^{(0)}_r}(W^{(k)}_r)- r \epsilon$ doesn't have  any root in $\epsilon=(0,1)$. Finally, one can choose an $\epsilon$ in the interval $(0,1)$ and compare the values of $\sum_{k=1}^r Z_{P^{(0)}_r}(W^{(k)}_r)$ and $ r\epsilon$ at that point to see that the capacity of proposed modified scheme is greater than the repetition one for $r$ number of repetitions. 
Fig. (\ref{fig:proof}) shows the left and the right sides of eq. (\ref{Theorem}) for $r=4$.
\begin{figure}[h]
\vspace{-4mm}
\centering
\includegraphics[width=1\linewidth]{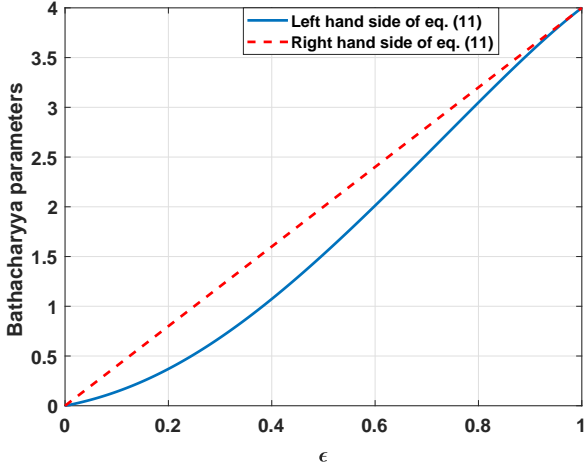}
\caption{Comparison between the left and the right hand sides of eq. (\ref{Theorem}).}
\label{fig:proof}
\end{figure}


\section{Analysis and Numerical Results} 
In this section, we first analyze the numerical result of the proposed scheme and compare it with the straightforward polar repetition scheme. Then, we provide complexity analysis of both of these schemes.

\subsection{Numerical Analysis}

In this subsection, we provide numerical results for the 
capacity of the proposed polar coded repetition scheme for different numbers of repetitions over BEC and compare them with the capacity of the simple repetition scheme and the Shannon bound. Fig. (\ref{fig:comparison2}) illustrates the capacities of the proposed schemes for $2$, $4$ and $8$ repetitions. It can be observed that the proposed scheme outperforms the simple repetition scheme for all of these repetitions. The irregular scheme also slightly outperforms the regular one for $4$ repetitions. On the other hand, as the number of repetitions increases, the gap to the Shannon bound increases as expected. 

\begin{figure}[h]
\vspace{-1mm}
\centering
\includegraphics[width=1.1\linewidth]{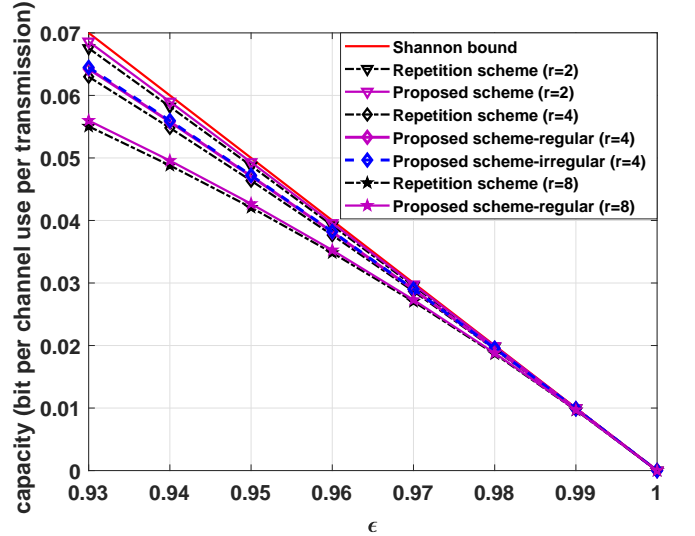}
\caption{Capacity of the proposed scheme compared with the capacity of the repetition scheme for $r=2,4,8$.}
\label{fig:comparison2}
\end{figure}

\subsection{Complexity Analysis} \label{Complexity}

The complexity of the straightforward  polar-repetition scheme consists of the complexity of the outer polar code of size $n$, $O(n \log n)$, and the repetition code of size $r$, $nr$. Thus, the complexity of the total decoding process for this scheme is $O(nr+ n \log n)$.

The complexity of the proposed polar coded repetition scheme consists of the complexity of the three stages. The first stage is $r$ different kernels $G^{\prime}_r$ where the complexity of each of $G^{\prime}_r$ of size $r$ is $ O(n \log r)$. The second stage is repetition code of size $r$  with complexity $nr$. Finally, the third stage is kernel $G_2^{\otimes (m-t)}$ of size $n/r$ with complexity $O({n} \log \frac{n}{r})$. Hence, the complexity of the total decoding process is $O(nr + rn \log r + {n} \log \frac{n}{r})=O(nr+n \log n+ n (r-1) \log r)$. If $r$ be a constant, then the complexity of the polar coded repetition will be of order $O(n \log n)$.  

\section{Conclusion}
In this paper, we proposed a modified approach for the repetition scheme. In this  scheme, we used polar codes as the outer code and proposed to transmit slightly modified codeword in each repetition. We showed that the proposed scheme outperforms the simple repetition scheme, in terms of the asymptotic achievable rate, over BEC while it keeps the decoding complexity almost the same as the repetition scheme.

\end{document}